\documentclass{article}

\usepackage{graphicx}
\usepackage{spconf,amsmath,graphicx}
\usepackage{amsthm}
\usepackage[ruled,vlined]{algorithm2e}
\usepackage{multirow}
\usepackage{array}
\usepackage[font=small,labelfont=bf]{caption}
\newcolumntype{P}[1]{>{\centering\arraybackslash}p{#1}}
 \usepackage{arydshln}
\usepackage{cite}
\usepackage{amssymb,amsfonts}
\usepackage{textcomp}
\usepackage{color}
\usepackage{subcaption}

\usepackage{fancyhdr}

\def\0{{\mathbf 0}}
\def\1{{\mathbf 1}}

\def\c{{\mathbf c}}

\def\e{{\mathbf e}}

\def\u{{\mathbf u}}
\def\v{{\mathbf v}}

\def\x{{\mathbf x}}
\def\y{{\mathbf y}}

\def\B{{\mathbf B}}
\def\C{{\mathbf C}}
\def\D{{\mathbf D}}
\def\E{{\mathbf E}}

\def\I{{\mathbf I}}

\def\L{{\mathbf L}}

\def\P{{\mathbf P}}
\def\Q{{\mathbf Q}}

\def\S{{\mathbf S}}
\def\T{{\mathbf T}}
\def\U{{\mathbf U}}
\def\V{{\mathbf V}}
\def\W{{\mathbf W}}

\def\Tr{{\text{Tr}}}

\def\ie{{\textit{i.e.}}}

\def\cE{{\mathcal E}}

\def\cG{{\mathcal G}}
\def\cH{{\mathcal H}}
\def\cI{{\mathcal I}}
\def\cK{{\mathcal K}}
\def\cL{{\mathcal L}}

\def\cN{{\mathcal N}}

\def\cV{{\mathcal V}}

\def\cX{{\mathcal X}}

\def\cN{{\mathcal N}}

\def\balpha{{\boldsymbol \alpha}}

\def\bLambda{{\boldsymbol \Lambda}}
\def\bOmega{{\boldsymbol \Omega}}

\newtheorem{theorem}{Theorem}[section]

\newtheorem{lemma}[theorem]{Lemma}

\def\ie{{\textit{i.e.}}}

\title{Learning Sparse Graph Laplacian with $k$ Eigenvector Prior via \\ 
Iterative GLASSO and Projection}
\name{Author(s) Name(s)\thanks{Thanks to XYZ agency for funding.}}
\address{Author Affiliation(s)}
\name{Saghar Bagheri$^{\star}$ \qquad Gene Cheung$^{\star}$ \qquad Antonio Ortega$^{\dagger}$
\qquad Fen Wang$^{\ddagger}$
\thanks{Gene Cheung acknowledges the support of the NSERC grants RGPIN-2019-06271,  RGPAS-2019-00110.}
}

\address{$^{\star}$ York University, Toronto, Canada ~~~~~~~~~~
$^{\ddagger}$ Xidian University, Xi'an, China \\
$^{\dagger}$ University of Southern California, Los Angeles, CA, USA
}

\begin{document}
\ninept
\maketitle
\begin{abstract}
Learning a suitable graph is an important precursor to many graph signal processing (GSP) pipelines, such as graph signal compression and denoising. 
Previous graph learning algorithms either i) make assumptions on graph connectivity (e.g., graph sparsity), or 
ii) make edge weight assumptions such as positive edges only.
In this paper, given an empirical covariance matrix $\bar{\C}$ computed from data as input, we consider an eigen-structural assumption on the graph Laplacian matrix $\L$:
the first $K$ eigenvectors of $\L$ are pre-selected, e.g., based on domain-specific criteria, and the remaining eigenvectors are then learned from data.  
One example use case is image coding, where the first eigenvector is pre-chosen to be constant, regardless of available observed data.
We first prove that the subspace $\cH_{\u}^+$ of symmetric positive semi-definite (PSD) matrices with the first $K$ eigenvectors being $\{\u_k\}$ in a defined Hilbert space is a convex cone.
We then construct an operator to project a given positive definite (PD) matrix $\L$ to $\cH_{\u}^+$, inspired by the Gram-Schmidt procedure.
Finally, we design an efficient hybrid graphical lasso / projection algorithm to compute the most suitable graph Laplacian matrix $\L^* \in \cH_{\u}^+$ given $\bar{\C}$. 
Experimental results show that given the first $K$ eigenvectors as a prior, our algorithm outperforms competing graph learning schemes using a variety of graph comparison metrics.
\end{abstract}
\section{Introduction}
\label{sec:intro}
\textit{Graph signal processing} (GSP) \cite{ortega18ieee} is the study of signals residing on graphs.
While GSP tools have been demonstrated to be effective in a wide range of applications from image compression and denoising to matrix completion \cite{cheung18,bai19,8976421dinesh,hu15,pang15,wang20}, a fundamental first step in GSP is the selection of an appropriate graph $\cG$ (or graph Laplacian matrix $\L$) that suitably describes pairwise (dis)similarities. 
Previous graph learning algorithms fall into two categories: 
i) \textit{connectivity}-based approaches like graphical lasso that, given an empirical covariance matrix $\bar{\C}$, estimate an inverse matrix with only sparse graph connectivity assumptions \cite{articleglasso,cai11_CLIME}; and 
ii) \textit{edge}-based approaches that assume edge properties such as edge signs and absence of self-loops\cite{egilmez17,cheung18tsipn,dong2016learning}.
Neither of these two categories of methods make assumptions about the \textit{eigen-structure} of the graph Laplacian matrix $\L$.

In this paper, we introduce explicit eigen-structure assumptions on the graph Laplacian matrix $\L$---that the first $K$ eigenvectors of $\L$ are pre-selected or computed based on domain-specific criteria---into a graph learning framework. 
Consider the example of image compression, where one can deduce from domain knowledge that the most common pixel block is the constant signal, and thus should be the first eigenvector regardless of training data. 
Consider also political voting records, where the most common pattern is voting along party affiliations in a two-party political system. 
Thus, the first eigenvector should be a piecewise constant signal (\ie, nodes of one party are assigned $1$, while nodes of the other are $-1$) \cite{Banerjee7}. 
There are also practical cases where the first $K$ eigenvectors can be pre-chosen for fast computation.
For example, \textit{fast graph Fourier transform} (FGFT)
\cite{lemagoarou2017analyzing} approximates an eigen-matrix by multiplying a small number of Givens rotation matrices. 
The first $K$ rows of the sparse matrix product are the $K$ approximate eigenvectors, and the $K$ transform coefficients can be computed speedily due to sparsity.

We propose an optimization strategy where a graph Laplacian matrix $\L$ is learned optimally from data, while restricting the first $K$ eigenvectors to be those chosen ahead of time. 
We first prove that the subspace $\cH_{\u}^+$ of symmetric \textit{positive semi-definite} (PSD) matrices with the first $K$
eigenvectors taken from a given set $\{\u_k\}^K_{k=1}$ is a convex cone. 
We next construct an operator to project a given positive definite (PD) matrix $\L$ to $\cH_{\u}^+$, inspired by the Gram-Schmidt procedure \cite{hoffmann1989iterative}. 
Finally, we design an efficient hybrid graphical lasso (GLASSO) / projection algorithm to compute the most suitable graph Laplacian matrix $\L^* \in \cH_{\u}^+$ given input empirical covariance matrix $\bar{\C}$.
Experimental results demonstrate that given the first $K$ eigenvectors as a prior, our algorithm outperformed competing graph learning methods in terms of different graph metrics.

\section{Preliminaries}
\label{sec:prelim}

\vspace{-0.05in}
\subsection{Graph Definitions}

\vspace{-0.05in}
Suppose we are given a graph $\cG(\cV,\cE)$ with $|\cV| = N$ nodes and edges $(i,j) \in \cE$ connecting nodes $i$ and $j$ with weight $w_{ij} \in \mathbb{R}$.
Denote by $\W$ the \textit{adjacency matrix}, where $W_{ij} = w_{ij}$.
$\W$ may include self-loops $W_{ii}$'s.
Assuming that the edges are undirected, $\W$ is symmetric.
Define next the diagonal \textit{degree matrix} $\D$ where $D_{ii} = \sum_{j} W_{ij}$. 
The \textit{combinatorial graph Laplacian matrix} \cite{ortega18ieee} is then defined as $\L = \D - \W$.
To properly account for self-loops, the \textit{generalized graph Laplacian matrix} is defined as $\cL = \D - \W + \text{diag}(\W)$.
In this paper, we assume that the target graph Laplacian to be learned from data may have self-loops, and edge weights can be positive or negative.
This means that our graph Laplacian $\cL$ can account for both pairwise correlation and anti-correlation. 

\vspace{-0.05in}
\subsection{Hilbert Space}

\vspace{-0.05in}
Following terminologies in \cite{vetterli2014foundations}, we first define a \textit{vector space} $\cX$ of real, symmetric matrices in $\mathbb{R}^{N \times N}$. 
Note that $\cX$ is closed under addition and scalar multiplication. 
We next define the standard \textit{inner product} $\langle \cdot, \cdot \rangle$ for matrices $\P, \Q$ in $\cX$ as:
\begin{align}
\langle \P, \Q \rangle &=
\Tr(\Q^{\top} \P) = \sum_{i,j} P_{ij} Q_{ij}.
\end{align}
This definition of inner product induces the Frobenius norm:
\begin{align}
\langle \P, \P \rangle &=
\| \P \|^2_F = \sum_{i,j} P_{ij}^2 \geq 0.
\end{align}
We can now define a \textit{Hilbert space} $\cH$ in $\cX$ as a space of real, symmetric matrices endowed with an inner product $\langle \cdot, \cdot \rangle: (\cH, \cH) \mapsto \mathbb{R}$.
Further, we are interested in the subspace of \textit{positive semi-definite} (PSD) matrices, 
\ie, $\cH^+ = \{\P \in \cH \,|\, \P \succeq 0\}$. 
One can easily show that subspace $\cH^+$ is a convex cone \cite{convex-optimization}.

\section{Projection Operator}
\label{sec:algo}

We first overview our algorithm development.
We prove that the subspace $\cH_{\u}^+ \subset \cH^+$ of symmetric PSD matrices 
that share a common \textit{ordered} set of first $K$ eigenvectors is a convex cone. 
This means that, given an empirical covariance matrix $\bar{\C}$ estimated from $M$ signal observations $\cX = \{\x_1, \ldots, \x_M\}$, minimizing a convex objective computed using $\bar{\C}$ and a graph Laplacian matrix $\L$ while restricting $\L \in \cH_{\u}^+$ is a convex optimization problem.
We then develop a projection operator to project a \textit{positive definite} (PD) matrix to $\cH_{\u}^+$. 
We describe our optimization algorithm using the projection operator in Section\;\ref{sec:algo2}.
\vspace{-0.1in}
\subsection{Subspace of Matrices with Common First $K$ Eigenvector}

Denote by $\cH_{\u}^+ \subset \cH^+$ the subspace of PSD matrices that share the first $K$ eigenvectors $\{\u_k\}_{k=1}^K$, assumed to be  orthonormal, \ie, 
\begin{align}
\u_j^{\top} \u_k = \delta_{j-k}, ~~~\forall j, k \in \cI_K
\end{align}
where $\cI_K = \{1, \ldots, K\}$ and $\delta_i$ is the discrete impulse function that evaluates to 1 if $i=0$ and 0 otherwise.
We can define $\cH_{\u}^+$ using the Rayleigh quotient and the min-max theorem \cite{golub12} as
\begin{align}
\cH_{\u}^+ = \left\{ \L \in \cH^+ \;|\;
\u_k = \arg \min_{\x \,|\, \x \perp \u_j, \forall j < k} \frac{\x^{\top} \L \x}{\x^{\top} \x}, ~~
k \in \cI_K
\right\}.
\end{align}
\vspace{-0.2in}
\begin{lemma}
$\cH_{\u}^+$ is a convex cone.
\end{lemma}

\begin{proof}
We prove by induction. 
For the base case when $K=1$, let $\L_1, \L_2 \in \cH^+$ be two matrices that share the first eigenvector $\u_1$. 
Consider a positive linear combination $\cL = c_1 \L_1 + c_2 \L_2$, where $c_1, c_2 \geq 0$.
We examine the smallest eigenvalue $\lambda_{\min}(\cL)$ of $\cL$ using the Rayleigh quotient again, \ie, 
\begin{align}
\lambda_{\min}(\cL) &= \min_{\x} \frac{\x^{\top} \cL \x}{\x^{\top} \x} 
= \min_{\x} \frac{\x^{\top} \left(c_1 \L_1 + c_2 \L_2 \right) \x}{\x^{\top} \x}
\nonumber \\
&\stackrel{(a)}{\geq} c_1 \left( \min_{\x} \frac{\x^{\top} \L_1 \x}{\x^{\top} \x} \right) +
c_2 \left( \min_{\y} \frac{\y^{\top} \L_2 \y}{\y^{\top} \y} \right)
\nonumber \\
&\stackrel{(b)}{=} c_1 \, \lambda_{\min}(\L_1) + 
c_2 \, \lambda_{\min}(\L_2) \nonumber
\end{align}
The inequality in $(a)$ is true since right-hand side (RHS) has more degrees of freedom than left-hand side (LHS) to minimize the two non-negative terms ($\L_1$ and $\L_2$ are PSD matrices in $\cH^+$).
$(b)$ is true by definition of Rayleigh quotient, with $\u_1$ being the minimizing argument for both terms by assumption. 
Thus, the argument that minimizes the Rayleigh quotient for $\cL$ is also $\u_1$, and therefore $\u_1$ is the first eigenvector of $\cL$. 
Since this analysis is true for all $c_1, c_2 \geq 0$, $\cH_{\u}^+$ is a convex cone for $K=1$. 

Consider next the inductive step, where we assume the $K=K_o$ case is true, and we prove the $K_o+1$ case.
Let $\L_1, \L_2 \in \cH^+$ be two matrices that share the first $K_o+1$ eigenvectors $\{\u_k\}_{k=1}^{K_o+1}$.
Consider a positive linear combination $\cL = c_1 \L_1 + c_2 \L_2$, where $c_1, c_2 \geq 0$.
By the inductive assumption, $\cL$ also shares the first $K_o$ eigenvectors $\{\u_k\}_{k=1}^{K_o}$.
For the $(K_o+1)$-th eigenvector, consider the $(K_o+1)$-th eigenvalue $\lambda_{K_o+1}(\cL)$, computed as
\begin{small}
\begin{align*}
&= \min_{\x|\x \perp \{\u_k\}^K_{k=1}} \frac{\x^{\top} \cL \x}{\x^{\top} \x} 
= \min_{\x|\x \perp \{\u_k\}^K_{k=1}} \frac{\x^{\top} \left( c_1 \L_1 + c_2 \L_2 \right) \x}{\x^{\top} \x} \\
&\geq c_1 \left( \min_{\x|\x \perp \{\u_k\}^K_{k=1}} \frac{\x^{\top} \L_1 \x}{\x^{\top} \x} \right) + 
c_2 \, \left( \min_{\y|\y \perp \{\u_k\}^K_{k=1}} \frac{\y^{\top} \L_2 \y}{\y^{\top} \y} \right) \\
&= c_1 \, \lambda_{K_o+1}(\L_1) + c_2 \lambda_{K_o+1}(\L_2).
\end{align*} 
\end{small}\noindent
Since the argument that minimizes the Rayleigh quotient for both $\L_1$ and $\L_2$ is $\u_{K_o+1}$, this is also the argument that minimizes the Rayleight quotient for $\cL$.
Thus the ($K_o+1$)-th eigenvector for $\cL$ is also $\u_{K_o+1}$. 
Thus, we conclude that $\cL$ has $\{\u_k\}_{k=1}^{K_o+1}$ as its first $K_o + 1$ eigenvectors.
Since both the base case and the inductive step are true, $\cH^+_{\u}$ is a convex cone.
\end{proof}

\vspace{-0.2in}
\subsection{Projection to Convex Cone $\cH_{\u}^+$}

To project a PD matrix $\P \in \cH^+$ to convex cone $\cH_{\u}^+$ given first $K$ eigenvectors $\{\u_k\}^K_{k=1}$, we approximate the inverse $\C = \P^{-1}$ with $\hat{\C}$ and ensure $\u_1$ is the argument that \textit{maximizes} the Rayleigh quotient of $\hat{\C}$, then $\u_2$ is the argument that maximizes the Rayleigh quotient of $\hat{\C}$ while being orthogonal to $\u_1$, and so on.
$\hat{\C}$ will be a linear combination of $\{\U_k\}^K_{k=1}$ and $\{\V_k\}^N_{k=K+1}$, where $\U_k = \u_k \u_k^{\top}$, $\V_k = \v_k \v_k^{\top}$, and $\v_k$'s are the $N-K$ remaining eigenvectors of $\hat{\C}$.

Specifically, we first compute 
$\C = \P^{-1}$. 
Define $\bOmega = \{\alpha \U_1, \alpha \in \mathbb{R}\}$ as the subspace spanned by rank-1 matrix $\U_1 = \u_1 \u_1^{\top}$ with scalar $\alpha$.
Denote by $\bOmega^{\perp}$ the orthogonal subspace in $\cH$ so that $\cH = \bOmega \oplus \bOmega^{\perp}$.
We first write $\C$ as its projection to $\bOmega$ plus its orthogonal component, \ie, 
\begin{align}
\C = \langle \C, \U_1 \rangle \U_1 
~+~ 
\C_{\bOmega^{\perp}} 
\label{eq:decompose}
\end{align}
where $\C_{\bOmega^{\perp}}$ is the orthogonal component of $\C$ in subspace $\bOmega^{\perp}$.
We show that given $\C$ is PD, inner product $\langle \C, \U_1 \rangle$ is non-negative.

\begin{lemma}
$\langle \C, \U_1 \rangle \geq 0$.
\end{lemma}

\begin{proof}
Since $\C$ is PD by assumption, without loss of generality we can write $\C = \B \B^{\top}$. 
Then
\begin{small}
\begin{align}
\langle \C, \U_1 \rangle &= \text{Tr} \left( ( \B \B^{\top})^{\top}  \u_1 \u_1^{\top} \right) 
= \text{Tr} \left( \B \B^{\top} \u_1 \u_1^{\top} \right) \nonumber \\
&= \text{Tr} \left( \B^{\top} \u_1 \u_1^{\top} \B \right) 
= \langle \u_1^{\top} \B, \u_1^{\top} \B \rangle.
\nonumber
\end{align}
\end{small}
The last term is a Frobenius norm $\| \u_1^{\top} \B \|^2_2 \geq 0$.
\end{proof}

Define $\mu_{1} = \langle \C, \U_1 \rangle \geq 0$. 
We express $\bOmega^{\perp}$ as the span of orthogonal rank-1 matrices $\U_2, \ldots, \U_{K}, \V_{K+1}, \ldots, \V_{N}$, where $\U_k = \u_k \u_k^{\top}$, $\V_i = \v_i \v_i^{\top}$, and $\|\v_i\|_2 = 1$. 
Suppose that the inner products of $\C$ with the orthogonal components, $\langle \C, \U_k \rangle, \langle \C, \V_i \rangle$, are within range $[0, \mu_{1}]$, $\forall k, i$. 
This means that the respective Rayleigh quotients are no larger than $\mu_1$:
\begin{align}
\text{Tr}(\u_k^{\top} \C \u_k) 
= \text{Tr}(\C \u_k \u_k^{\top})
= \langle \C, \U_k \rangle \leq \mu_1, 
~~\forall k
\end{align}
Then $\u_1$ is surely the \textit{last} eigenvector of $\C$ corresponding to \textit{largest} eigenvalue $\mu_{1}$. 
If this is not the case, then we need to approximate orthogonal component $\C_{\bOmega^{\perp}}$ as $\hat{\C}_{\bOmega^{\perp}}$ in order to ensure $\u_1$ is indeed the last eigenvector:
\begin{small}
\begin{align}
\min_{\{\V_i\}, \hat{\C}_{\bOmega^{\perp}}} 
\| \C_{\bOmega^{\perp}} - \hat{\C}_{\bOmega^{\perp}} \|^2_2, ~~~\mbox{s.t.} ~~
\left\{ \begin{array}{l}
\langle \V_i, \U_k \rangle = 0, ~\forall i, k \\
\langle \V_i, \V_j \rangle = \delta_{i-j}, ~\forall i, j \\
\V_i = \v_i \v_i^{\top}, ~\forall i \\
\|\v_i \|_2 = 1, ~\forall i \\
0 \leq \langle \hat{\C}_{\bOmega^{\perp}}, \U_k \rangle \leq \mu_{1}, ~\forall k \\
0 \leq \langle \hat{\C}_{\bOmega^{\perp}}, \V_i \rangle \leq \mu_{1}, ~\forall i
\end{array} \right.
\label{eq:origOpt}
\end{align}
\end{small}
Jointly optimizing $N-K$ rank-1 matrices $\V_i$'s and $\hat{\C}_{\Omega^{\perp}}$ in 
\eqref{eq:origOpt} is difficult. 
We propose a greedy algorithm instead next.

\vspace{-0.08in}
\subsection{Gram-Schmidt-inspired Algorithm}

Our algorithm is inspired by the \textit{Gram-Schmidt} procedure \cite{hoffmann1989iterative} that iteratively computes a set of orthonormal vectors given a set of linearly independent vectors. 
Similarly, our algorithm iteratively computes one orthogonal rank-1 matrix $\V_t = \v_t \v_t^{\top}$ at each iteration $t$, for $t \geq K+1$.
For iteration $t \leq K$, first $K$ eigenvector $\{\u_k\}^K_{k=1}$ are known, so we will simply use $\U_t = \u_t \u_t^{\top}$.

Consider iteration $t=2$.
From \eqref{eq:decompose}, we first define the \textit{residual signal} as $\E_1 = \C - \langle \C, \U_1 \rangle \U_1$.
Consider first the case where $K \geq 2$, and thus rank-1 matrix $\U_2 = \u_2 \u_2^{\top}$ orthogonal to $\U_1$ is pre-selected. 
Thus, we require only that the inner product $\mu_2 = \langle \hat{\C}, \U_2 \rangle$ of approximation $\hat{\C}$ and $\U_2$ to be $\leq \mu_1$. 
Mathematically,
\begin{align}
\mu_2 = \left\{ \begin{array}{ll}
\langle \E_1, \U_2 \rangle, & \mbox{if}~~ \langle \E_1, \U_2 \rangle \leq \mu_1 \\
\mu_1, & \mbox{o.w.}
\end{array} \right.
\end{align}

Consider next the case where $K=1$, and we must identify a new rank-1 matrix $\V_2 = \v_2 \v_2^{\top}$ orthogonal to $\U_1$ to reconstruct $\hat{\C}$.
In this case, to minimize error $\|\C - \hat{\C}\|^2_2$, we seek a $\V_2$ ``most aligned" with $\E_1$, \ie, 
\begin{align}
\max_{\v_2} ~ \langle \E_1, \v_2 \v_2^{\top} \rangle, 
~~\mbox{s.t.} ~~
\left\{ \begin{array}{l}
\langle \v_2 \v_2^{\top}, \U_1 \rangle = 0 \\
\|\v_2 \|_2 = 1
\end{array} \right.
\label{eq:maxEnergy}
\end{align}
Essentially, \eqref{eq:maxEnergy} seeks a \textit{rank-1 approximation} $\langle \E_1, \V_2 \rangle \V_2$ of matrix $\E_1$, while constraining $\V_2$ to be orthogonal to $\U_1$. 
The objective is equivalent to $\v_2^{\top} \E_1 \v_2$, which is convex for a maximization problem.

To convexify the problem, we first approximate $\E_1 \approx \e \e^{\top}$ where $\e$ is the last eigenvector\footnote{Extreme eigenvectors of sparse symmetric matrices can be computed efficiently using fast numerical methods such as LOBPCG \cite{lobpcg}.} of $\E_1$, and $\e \e^{\top}$ is the best rank-1 approximation of $\E_1$.
We then relax the norm equality constraint and rewrite \eqref{eq:maxEnergy} into the following optimization:
\begin{align}
\max_{\v_2} \e^{\top} \v_2, 
~~\mbox{s.t.} ~~ 
\left\{ \begin{array}{l}
\v_2^{\top} \u_1 = 0 \\
\| \v_2 \|^2_2 \leq 1
\end{array} \right.
\label{eq:maxEnergy2}
\end{align}
Optimization \eqref{eq:maxEnergy2} is now convex and solvable in polynomial time, using algorithms such as \textit{proximal gradient} (PG) \cite{PG}.

Having computed $\V_2 = \v_2 \v_2^{\top}$ in \eqref{eq:maxEnergy2}, we project $\E_1$ onto $\V_2$ and threshold its inner product to within $[0, \mu_1]$ as done in the $K \geq 2$ case, \ie,
\begin{align}
\mu_2 = \left\{ \begin{array}{ll}
\langle \E_1, \V_2 \rangle, & \mbox{if} ~ \langle \E_1, \V_2 \rangle \leq \mu_{1} \\
\mu_1, & \mbox{o.w.} 
\end{array} \right.
\label{eq:threshold}
\end{align}
Note that $\langle \E_1, \V_2 \rangle \geq 0$.
We omit the proof of the more general case $\langle \E_k , \V_k \rangle \geq 0$ for brevity.
The projection of $\E_1$ on $\V_2$ is thus $\mu_2 \V_2$.
We then compute new residual signal $\E_2$:
\begin{align}
\E_2 = \E_1 - \mu_2 \V_2.     
\end{align}
Given residual $\E_2$, we again compute a rank-1 matrix $\V_3 = \v_3 \v_3^{\top}$---orthogonal to $\V_2$ and $\U_1$---that is most aligned with $\E_2$, and compute projection of $\E_2$ to $\V_3$, and so on. 

More generally, at each iteration $t \leq K$, we threshold the inner product $\mu_{t+1} = \min (\langle \E_t, \U_{t+1} \rangle, \mu_t)$, and compute the next residual signal $\E_{t+1} = \E_t - \mu_{t+1} \U_{t+1}$.
On the other hand, at each iteration $t \geq K+1$, we first approximate $\E_t \approx \e \e^{\top}$ where $\e$ is $\E_t$'s last eigenvector. 
We then compute an optimal $\v_{t+1}$:
\begin{small}
\begin{align}
\max_{\v_{t+1}} ~ \e^{\top} \v_{t+1}, 
~~\mbox{s.t.} ~~ 
\left\{ \begin{array}{l}
\v_{t+1}^{\top} \u_k = 0, ~~ k \in \cI_K \\
\v_{t+1}^{\top} \v_{\tau} = 0, ~~ \tau \in \{K+1, \ldots, t\} \\
\| \v_{t+1} \|_2^2 \leq 1
\end{array} \right.
\label{eq:maxEnergy3}
\end{align}
\end{small}\noindent

We threshold the inner product $\mu_{t+1} = \min (\langle \E_t, \V_{t+1} \rangle, \mu_t)$, where $\V_{t+1} = \v_{t+1} \v_{t+1}^{\top}$. 
We compute the next residual signal $\E_{t+1} = \E_t - \mu_{t+1} \V_{t+1}$.

We finally note that our constructed operator $\text{Proj}(\cdot)$ is indeed a projection, since it is provably \textit{idempotent} \cite{vetterli2014foundations}, \ie, $\text{Proj}(\text{Proj}(\P)) = \text{Proj}(\P)$ for any PD matrix $\P \in \cH^+$. 

\section{Graph Laplacian Matrix Estimation}
\label{sec:algo2}


Given an empirical covariance matrix $\bar{\C}$ computed from data, we formulate the following optimization problem to estimate a graph Laplacian matrix $\L$ starting from the GLASSO formulation in \cite{mazumder2012graphical}:
\begin{align}
\min_{\L \in \cH^+_{\u}} ~~ \text{Tr}(\L \bar{\C}) -\log \det \L + \rho \; \| \L \|_1
\label{eq:glasso}
\end{align}
where $\rho > 0$ is a shrinkage parameter for the $l_1$ norm.
The only difference from GLASSO is that \eqref{eq:glasso} has an additional constraint $\L \in \cH^+_{\u}$. 
Because $\cH^+_{\u}$ is a convex set, \eqref{eq:glasso} is a convex problem.


We solve \eqref{eq:glasso} iteratively using our developed projection operator in Section\;\ref{sec:algo} and a variant of the \textit{block Coordinate descent} (BCD) algorithm in \cite{wright2015coordinate}. 
Specifically, we first solve the \textit{dual} of GLASSO as follows.
We first note that the $l_1$ norm in \eqref{eq:glasso} can be written as 
\begin{align}
  \|\L\|_1 = \max_{\| \U \|_\infty\leq 1} ~~\text{Tr}(\L\U) 
\end{align}
where $\| \U \|_\infty$ is the maximum absolute value element of the symmetric matrix $\U$. 
Then the dual problem of GLASSO that solves for an estimated covariance matrix $\C = \L^{-1}$ is
\begin{align}
\min_{\C^{-1} \in \cH^+} ~~ -\log \det \C, 
~~~\mbox{s.t.} ~~ 
\| \C - \bar{\C} \|_\infty \leq \rho
\label{eq:glasso_dual}    
\end{align}
where $\C =\bar{\C}+\U$  implies that the primal and dual variables are related via $\L = {(\bar{\C}+\U)}^{-1} $ \cite{Banerjee7}. 
To solve \eqref{eq:glasso_dual}, 
we update one row-column pair in $\C$ in \eqref{eq:glasso_dual} in each iteration. 
Specifically, we first reorder the rows / columns of $\C$ so that the optimizing row and column are swapped to the last.
We then partition $\bar{\C}$ and $\C$ into blocks: 
\begin{align}
\C = 
\begin{pmatrix}
    \C_{11} & \c_{12}\\
    \c_{12}^\top & c_{22}
\end{pmatrix},~~ 
\bar{\C} =
\begin{pmatrix}
    \bar{\C}_{11} & \bar{\c}_{12}\\
    \bar{\c}_{12}^\top & \bar{c}_{22}
\end{pmatrix}
\end{align}\noindent
\cite{Banerjee7} showed that the diagonals of $\C$ remain unchanged, and the optimal $\c_{12}$ is the solution to the following linearly constrained quadratic programming problem:
\begin{align}
\c_{12} = \text{arg}\,\min\limits_{\y}\, \{\y^\top\C_{11}^{-1}\y\},  ~~\mbox{s.t.} ~~ \| \y - \bar{\c}_{12} \|_\infty \leq \rho
\label{eq:c12}
\end{align}

Our algorithm to solve \eqref{eq:glasso} is thus as follows.
We minimize the GLASSO terms in \eqref{eq:glasso} by solving its dual \eqref{eq:glasso_dual}---iteratively updating one row / column of $\C$ using \eqref{eq:c12}.
We then project $\L = \C^{-1}$ to convex cone $\cH_{\u}^+$ using our projection operator. 
We repeat these two steps till convergence.




\begin{figure*}[!ht]
\centering
\subfloat[]{\includegraphics[width=0.18\textwidth]{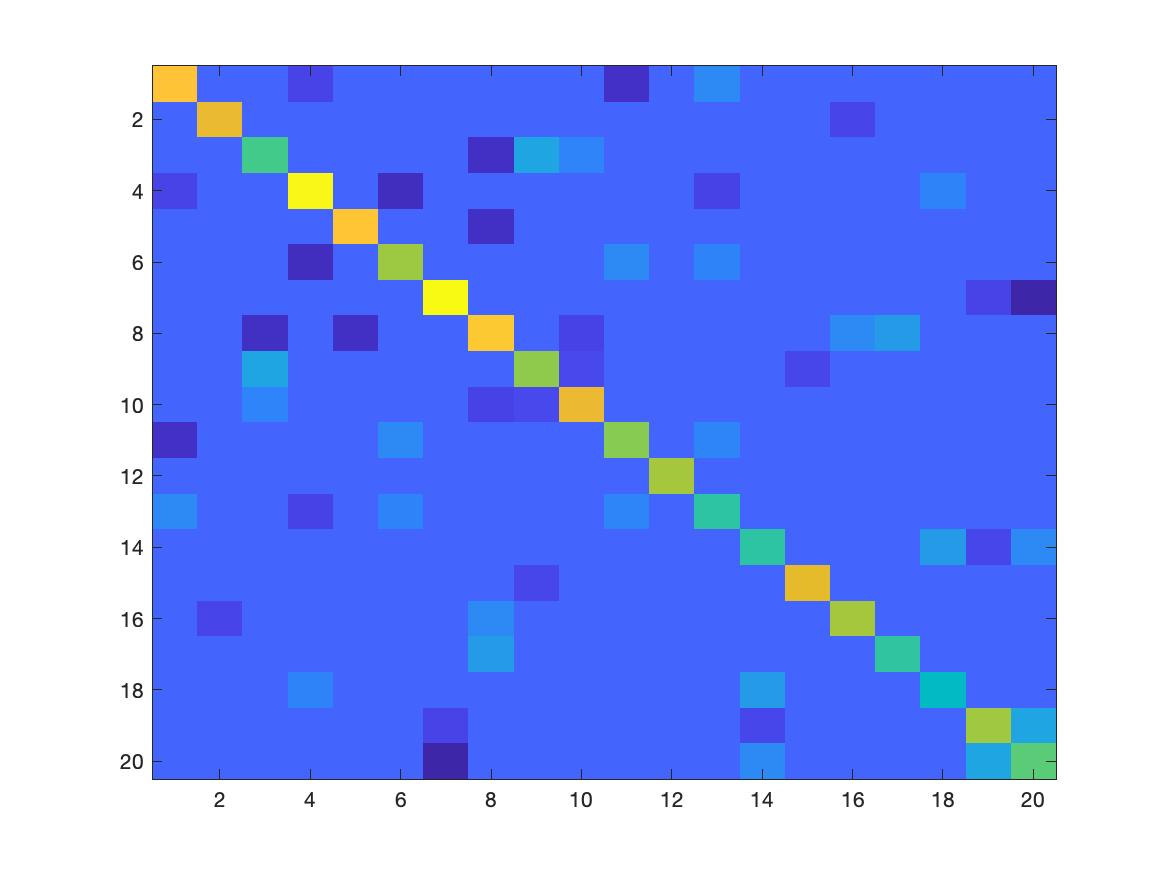}}
\subfloat[]{\includegraphics[width=0.18\textwidth]{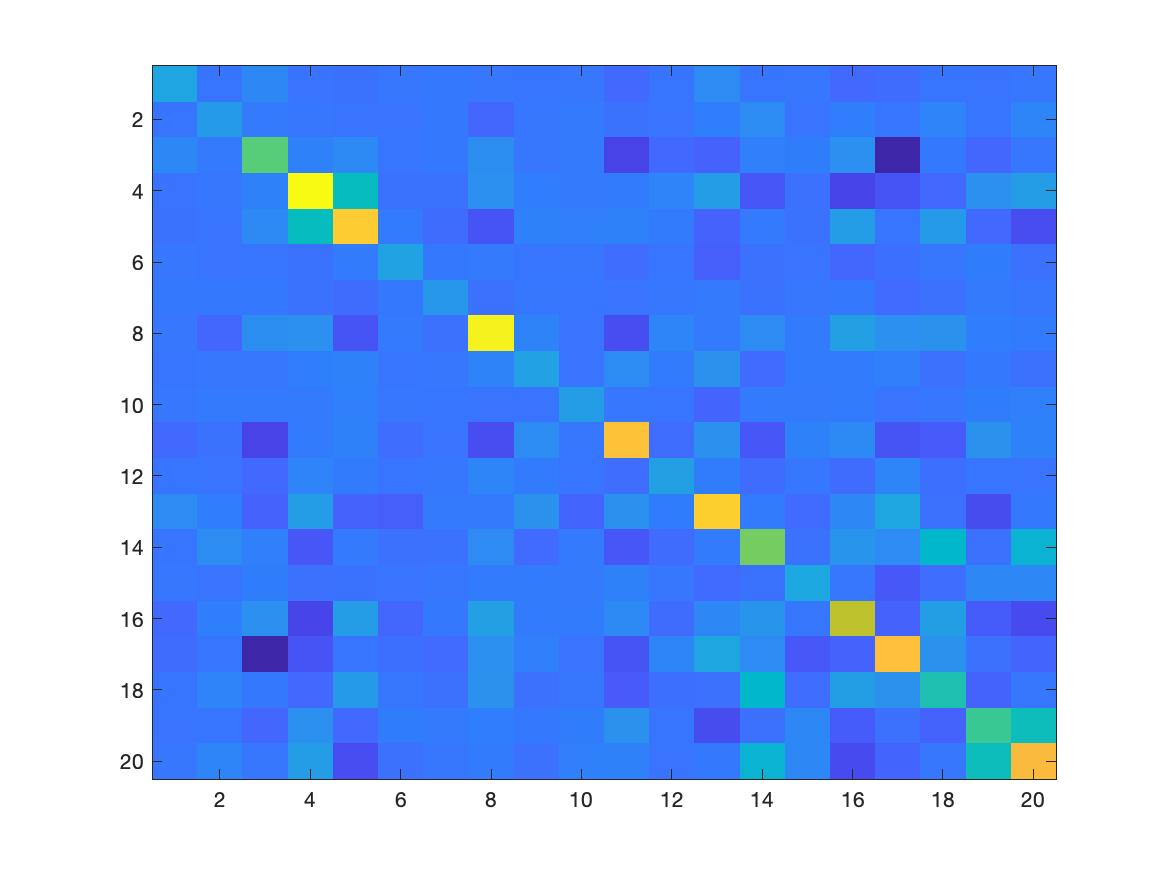}}
\subfloat[]{\includegraphics[width=0.18\textwidth]{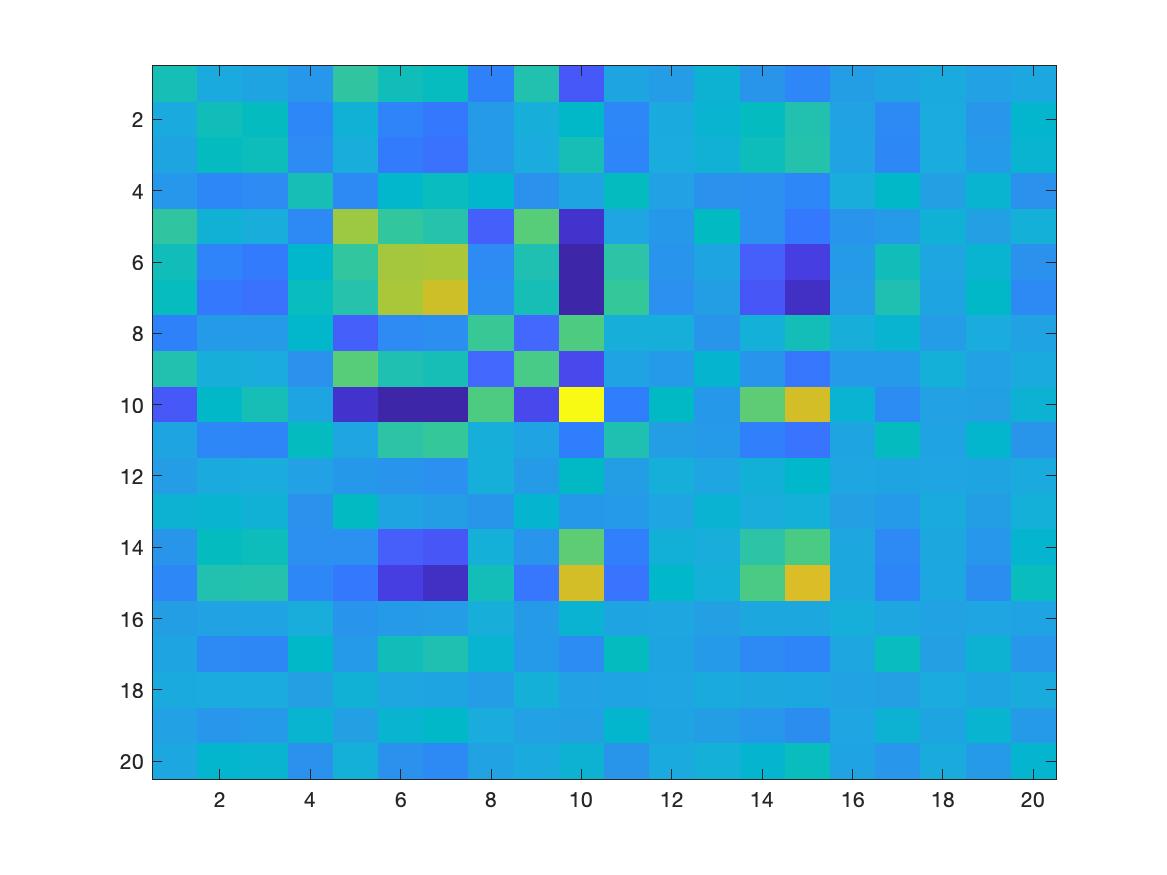}}
\subfloat[]{\includegraphics[width=0.18\textwidth]{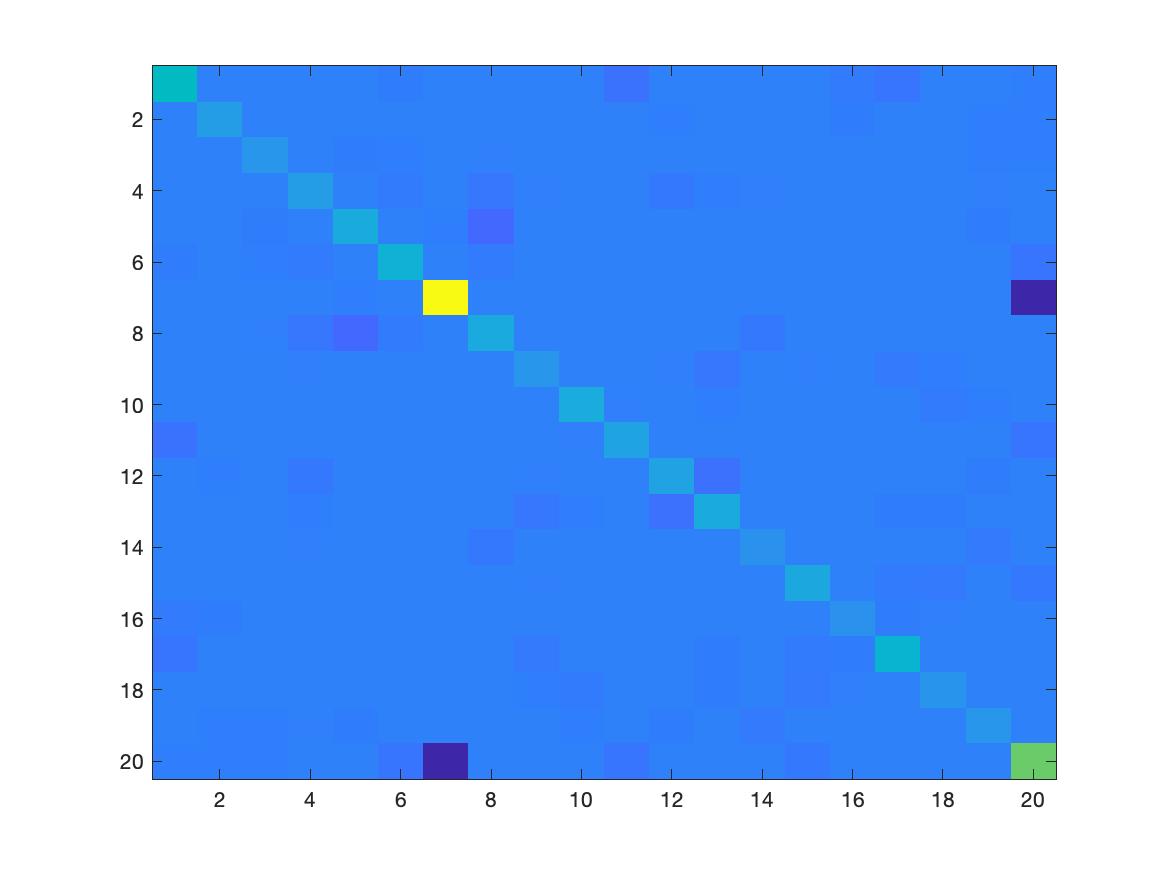}}
\subfloat[]{ \includegraphics[width=0.18\textwidth]{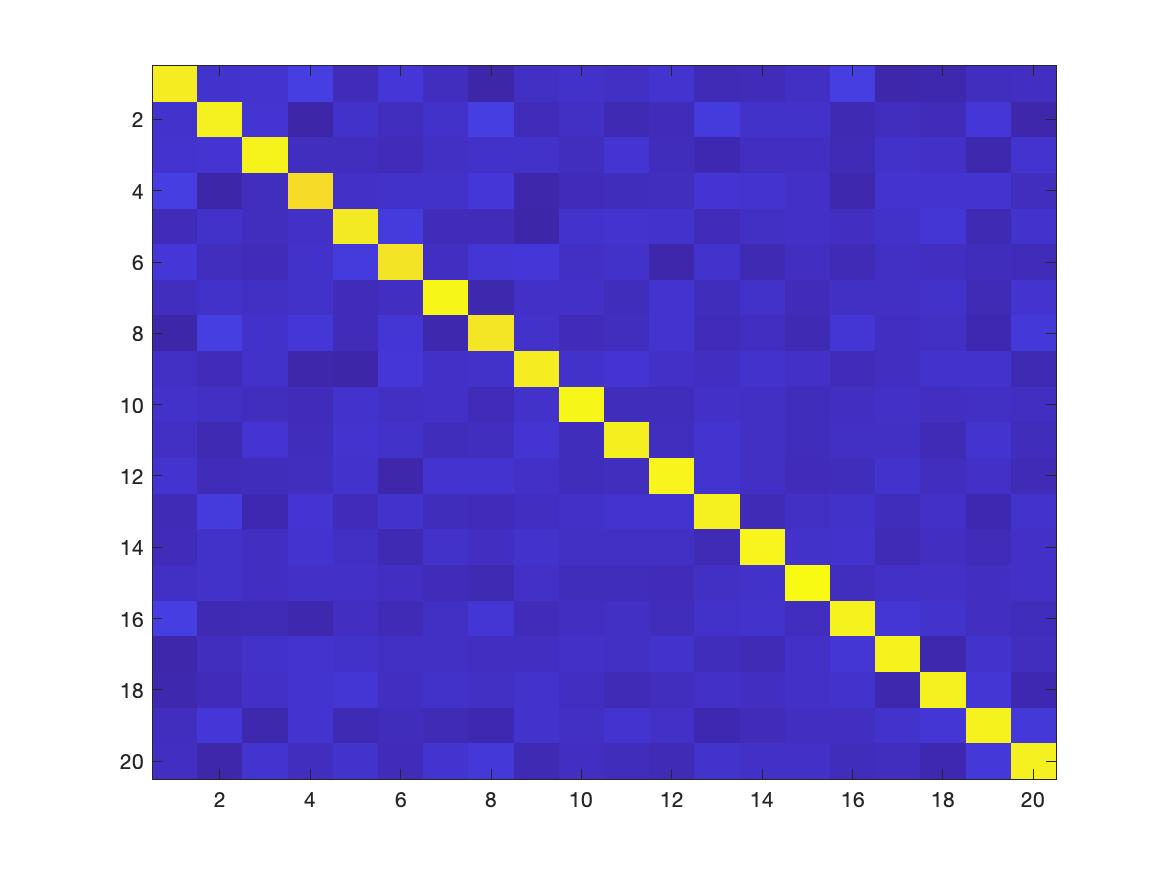}}
\vspace{-0.15in}
    \caption{\textbf{(a)} Ground Truth Laplacian $\cL$ , \textbf{(b)} Proposed Proj-Lasso with $K=1$,  \textbf{(c)} GLASSO  \cite{articleglasso}, \textbf{(d)} DDGL \cite{egilmez17} and \textbf{(e)} GL-SigRep \cite{dong2016learning} .}\label{fig:gl}
\end{figure*}

\vspace{-0.15in}

\section{Experimental Results}
\label{sec:results}
\vspace{-0.05in}
\subsection{Experiments Setups for Synthetic Datasets}

\vspace{-0.05in}
We conducted experiments using synthetic datasets to evaluate the performance of our method (Proj-Lasso) and of competing schemes: 
Graphical Lasso (GLASSO) \cite{articleglasso}, graph Laplacian learning with Gaussian probabilistic prior ({GL-SigRep})  \cite{dong2016learning} and diagonally dominant generalized graph Laplacian estimation under structural and Laplacian constraints ({DDGL}) \cite{egilmez17}.
The convergence tolerance for those  algorithms was set to $\epsilon = 10^{-4}$, and the regularization parameter was $\rho = e^{-6}$. 

To simulate ground truth graphs, we first randomly located 20 nodes in 2D space and used the Erdos-Renyi model \cite{erdos59a} to determine their connectivity with probability $0.6$. 
We then computed edge weights using a Gaussian kernel, \textit{i.e.}, $w_{ij} = \exp{(-d(i, j)^2 /2\sigma^2)}$, where $d(i, j)$ is the Euclidean distance between $i$ and $j$ and $\sigma$ is 0.5. 
Edge weights smaller than $0.75$ were removed for sparsity. 
To introduce negative weights, we flipped the sign of each edge with probability 0.5 and then computed the generalized graph Laplacian $\cL$ (with no self-loops). 
To generate data from $\cL$, we first computed covariance $\cK = (\cL + \epsilon \I)^{-1}$ for $\epsilon = 0.5$. 
We then generated data set $\cX = \{\mathbf{x}_i\}_{ i=1}^{20}$ from multivariate normal distribution  $\mathbf{x} \sim \cN (0,\cK)$.
An empirical covariance matrix $\bar{\C}$ was then computed from $\cX$ and used as input to different graph learning algorithms.

We employed three popular graph similarity metrics to evaluate graph Laplacian matrices learned: 
relative error (RE), DeltaCon and $\lambda$-distance \cite{DeltaCon,metrics,lamda-dist-bunke2006graph,lamda-dist}. 
Specifically, RE computes the relative \textit{Frobenius norm} error between the ground truth Laplacian matrix $\cL$ and the learned matrix  $\hat{\L}$. \textit{DeltaCon} compares the similarities between all node pairs in the two graphs. 
$\lambda$-distance metric computes the eigenvalues' distance between two matrices.


\vspace{-0.1in}
\begin{table}[h!] 
\centering
 \resizebox{0.75\textwidth}{!}{\begin{minipage}{\textwidth}
  \begin{tabular}{|P{1.4cm}|P{1.3cm}|P{1.3cm}|P{1.3cm}|P{0.9cm}|P{0.9cm}|P{0.9cm}|}
    \hline
    \multirow{2}{*}{Metric} &
     \multirow{2}{*}{GLASSO} &
      \multirow{2}{*}{GL-SigRep} &
      \multirow{2}{*}{DDGL} &
      \multicolumn{3}{c|}{Proj-Lasso with different $K$'s}
      \\ \cline{5-7}
        &&&&$1$ & $2$ & $3$ \\
    \hline
    RE & 0.9306 & 0.7740 &  0.7543 & 0.7295 & 0.5191 & 0.4005   \\
    \hline
    DeltaCon & 0.9422 & 0.9449 & 0.9407 & 0.9502 & 0.9495 & 0.9586  \\
    \hline
    $\lambda$-distance & 1e+03 & 8.7168 & 28.7918 & 9.0983 & 3.3768 & 2.7125\\
    \hline
  \end{tabular}
  \end{minipage}}
  
\vspace{-0.1in}
\caption{Average RE, Deltacon and $\lambda$-distance for graph produced with 20 vertices and 20 signals on each node.}   \label{tab:error}
\end{table}

\vspace{-0.2in}
\begin{table}[h!] \centering
  \begin{tabular}{|P{1.4cm}|P{1.3cm}|P{1cm}|P{1cm}|P{1cm}|}
    \hline
    \multirow{1}{*}{Metric} &
    $K=N $&
     $K=1$ &
      $K=2$ &
     $K=3$ \\
    \hline
    RE & 1.4224 & 1.3375 &  1.3595 & 1.3595  \\
    \hline
    DeltaCon & 0.9283 & 0.9295 & 0.9237 & 0.9323\\
    \hline
    $\lambda$-distance & 2.5390 & 2.3943 & 2.4010 & 2.4010\\
    \hline
  \end{tabular}
\vspace{-0.1in}  
\caption{Average Relative Errors using different number of first $K$ eigenvectors from Givens rotation method \cite{lemagoarou2017analyzing}.}
  \label{table2}
\end{table}

\vspace{-0.2in}
\subsection{Results assuming the First $K$ Eigenvectors}

\vspace{-0.05in}
Table\;\ref{tab:error} shows the graph learning performance of different methods evaluated using the aforementioned three metrics. 
For our proposed Proj-Lasso, we also simulated its performance using different $K$'s---the number of known first eigenvectors---which in our experiments were the first $K$ eigenvectors of the ground truth matrix $\cL$. 
As shown in Table\;\ref{tab:error}, using RE as metric, Proj-Lasso outperformed GLASSO, DDGL and GL-SigRep when $K=1$ and became even better as $K$ increased. 
Using $\lambda$-distance as metric, Proj-Lasso was always better than  its competitors. 
Using DeltaCon as metric, all the methods are comparable.  
We observe that Proj-Lasso's performance evaluated using RE and $\lambda$-distance improved significantly as $K$ increased, while the results of DeltaCon metric were not sensitive to $K$. 

We show visual comparisons of learned Laplacian matrices in Fig.\;\ref{fig:gl}, demonstrating that Proj-Lasso had the best performance, \textit{i.e.}, Proj-Lasso is visually closer to the ground truth matrix than others. 

\vspace{-0.1in}
\subsection{Results using Eigenvectors from Givens Rotation Method}

\vspace{-0.05in}
\cite{lemagoarou2017analyzing} developed a method based on Givens rotation matrices to approximately diagonalize Laplacian $\L$, \ie, 
\begin{align}
    \L \approx \S_1 \cdots \S_J \hat{\bLambda} \S_J^\top \cdots \S_1^\top
\end{align}
where $\S_1, \ldots , \S_J$ are Givens rotation matrices that are
both sparse and orthogonal, and $\hat{\bLambda}$ is a near-diagonal matrix.
$\T = \S_J^\top \cdots \S_1^\top$ can be interpreted as a fast \textit{Graph Fourier Transform} (FGFT) that approximates the original eigen-matrix of $\L$.
$\T$ is sparse since each $\S_j$ is sparse, and thus computation of transform coefficients $\balpha = \T \x$ can be efficient.
$J$ is a parameter to trade off the complexity of the transform $\T$ and the GFT approximation error.

In this experiment, we assumed that the first $K$ rows of $\T$ are chosen as the first $K$ eigenvectors in our prior, which are the input to our algorithm Proj-Lasso.
The algorithm then computed the remaining $N-K$ eigenvectors to compose Laplacian $\L$.
We set $J=2000$.
We compared the performance of Proj-Lasso against the scheme that uses all the $N$ rows of $\T$. Table\;\ref{table2} shows learning performances using different numbers of rows from $T$ ($K=1, \ldots 3$ is Proj-Lasso, and $K=N$ is Givens method). 
We observe that Proj-Lasso has smaller error for both metrics RE and $\lambda$-distance. 


\section{Conclusion}
\label{sec:conclude}
Given observable graph signals to compute an empirical covariance matrix $\bar{\C}$, we propose a new graph learning algorithm to compute the most likely graph Laplacian matrix $\L$, where we assume that the first $K$ eigenvectors of $\L$ are pre-selected based on domain knowledge or an alternative criterion. 
We first prove that the subspace $\cH^+_{\u}$ of symmetric positive semi-definite (PSD) matrices sharing the first $K$ eigenvectors $\{\u_k\}$ in a defined Hilbert space is a convex cone. 
We construct an operator, inspired by the Gram-Schmidt procedure, to project a positive definite (PD) matrix $\P$ into $\cH^+_{\u}$.
We design an efficient algorithm to compute the most suitable $\L \in \cH^+_{\u}$ given $\bar{\C}$, combining block coordinate descent (BCD) in GLASSO and our projection operator. 
Experimental results show that our algorithm outperformed competing graph learning schemes when the first $K$ eigenvectors are known.




\newpage
\bibliographystyle{IEEEbib}
\bibliography{ref2}

\begin{thebibliography}{10}

\bibitem{ortega18ieee}
A.~Ortega, P.~Frossard, J.~Kovacevic, J.~M.~F. Moura, and P.~Vandergheynst,
\newblock ``Graph signal processing: Overview, challenges, and applications,''
\newblock in {\em Proceedings of the {IEEE}}, May 2018, vol. 106, no.5, pp.
  808--828.

\bibitem{cheung18}
G.~Cheung, E.~Magli, Y.~Tanaka, and M.~Ng,
\newblock ``Graph spectral image processing,''
\newblock in {\em Proceedings of the {IEEE}}, May 2018, vol. 106, no.5, pp.
  907--930.

\bibitem{bai19}
Y.~Bai, G.~Cheung, X.~Liu, and W.~Gao,
\newblock ``Graph-based blind image deblurring from a single photograph,''
\newblock in {\em IEEE Transactions on Image Processing}, March 2019, vol. 28,
  no.3, pp. 1404--1418.

\bibitem{8976421dinesh}
C.~{Dinesh}, G.~{Cheung}, and I.~V. {Bajić},
\newblock ``Point cloud denoising via feature graph {Laplacian}
  regularization,''
\newblock {\em IEEE Transactions on Image Processing}, vol. 29, pp. 4143--4158,
  2020.

\bibitem{hu15}
W.~Hu, G.~Cheung, A.~Ortega, and O.~Au,
\newblock ``Multi-resolution graph {Fourier} transform for compression of
  piecewise smooth images,''
\newblock in {\em IEEE Transactions on Image Processing}, January 2015, vol.
  24, no.1, pp. 419--433.

\bibitem{pang15}
J.~Pang, G.~Cheung, A.~Ortega, and O.~C. Au,
\newblock ``Optimal graph {Laplacian} regularization for natural image
  denoising,''
\newblock in {\em IEEE International Conference on Acoustics, Speech and Signal
  Processing}, Brisbane, Australia, April 2015.

\bibitem{wang20}
F.~{Wang}, Y.~{Wang}, G.~{Cheung}, and C.~{Yang},
\newblock ``Graph sampling for matrix completion using recurrent {Gershgorin}
  disc shift,''
\newblock {\em IEEE Transactions on Signal Processing}, vol. 68, pp.
  2814--2829, 2020.

\bibitem{articleglasso}
J.~Friedman, T.~Hastie, and R.~Tibshirani,
\newblock ``Sparse inverse covariance estimation with the graphical lasso,''
\newblock {\em Biostatistics (Oxford, England)}, vol. 9, pp. 432--41, 08 2008.

\bibitem{cai11_CLIME}
T.~Cai, W.~Liu, and X.~Luo,
\newblock ``A constrained $\ell_1$ minimization approach to sparse precision
  matrix estimation,''
\newblock in {\em Journal of the American Statistical Association}, 2011, vol.
  106, pp. 594--607.

\bibitem{egilmez17}
H.~Egilmez, E.~Pavez, and A.~Ortega,
\newblock ``Graph learning from data under {Laplacian} and structural
  constraints,''
\newblock in {\em IEEE Journal of Selected Topics in Signal Processing}, July
  2017, vol. 11, no.6, pp. 825--841.

\bibitem{cheung18tsipn}
G.~Cheung, W.-T. Su, Y.~Mao, and C.-W. Lin,
\newblock ``Robust semisupervised graph classifier learning with negative edge
  weights,''
\newblock in {\em IEEE Transactions on Signal and Information Processing over
  Networks}, December 2018, vol. 4, no.4, pp. 712--726.

\bibitem{dong2016learning}
X.~{Dong}, D.~{Thanou}, P.~{Frossard}, and P.~{Vandergheynst},
\newblock ``Learning {Laplacian} matrix in smooth graph signal
  representations,''
\newblock {\em IEEE Transactions on Signal Processing}, vol. 64, no. 23, pp.
  6160--6173, 2016.

\bibitem{Banerjee7}
O.~Banerjee and L.~Ghaoui,
\newblock ``Model selection through sparse max likelihood estimation,''
\newblock {\em Journal of Machine Learning Research}, vol. 9, 08 2007.

\bibitem{lemagoarou2017analyzing}
L.~{Le Magoarou}, N.~{Tremblay}, and R.~{Gribonval},
\newblock ``Analyzing the approximation error of the fast graph {Fourier}
  transform,''
\newblock in {\em 2017 51st Asilomar Conference on Signals, Systems, and
  Computers}, 2017, pp. 45--49.

\bibitem{hoffmann1989iterative}
W.~Hoffmann,
\newblock ``Iterative algorithms for {Gram-Schmidt} orthogonalization,''
\newblock {\em Computing}, vol. 41, no. 4, pp. 335--348, 1989.

\bibitem{vetterli2014foundations}
M.~Vetterli, J.~Kova{v{c}}evi{\'c}, and V.K. Goyal,
\newblock {\em Foundations of Signal Processing},
\newblock Cambridge University Press, 2014.

\bibitem{convex-optimization}
S.~Boyd and L.~Vandenberghe,
\newblock {\em Convex optimization},
\newblock Cambridge university press, 2004.

\bibitem{golub12}
G.~Golub and C.~F.~Van Loan,
\newblock {\em Matrix Computations (Johns Hopkins Studies in the Mathematical
  Sciences)},
\newblock Johns Hopkins University Press, 2012.

\bibitem{lobpcg}
A.V. Knyazev,
\newblock ``Toward the optimal preconditioned eigensolver: Locally optimal
  block preconditioned conjugate gradient method,''
\newblock {\em SIAM journal on scientific computing}, vol. 23, no. 2, pp.
  517--541, 2001.

\bibitem{PG}
N.~Parikh and S.~Boyd,
\newblock ``Proximal algorithms,''
\newblock {\em Foundations and Trends in Optimization}, vol. 1, no. 3, pp.
  127--239, 2014.

\bibitem{mazumder2012graphical}
R.~Mazumder and T.~Hastie,
\newblock ``The graphical lasso: New insights and alternatives,''
\newblock {\em Electron. J. Statist.}, vol. 6, pp. 2125--2149, 2012.

\bibitem{wright2015coordinate}
S.J. Wright,
\newblock ``Coordinate descent algorithms,''
\newblock {\em Math. Program.}, vol. 151, no. 1, pp. 3--34, 2015.

\bibitem{erdos59a}
P.~Erd\"{o}s and A.~R\'{e}nyi,
\newblock ``On random graphs i,''
\newblock {\em Publicationes Mathematicae Debrecen}, vol. 6, pp. 290, 1959.

\bibitem{DeltaCon}
D.~Koutra, J.T. Vogelstein, and C.~Faloutsos,
\newblock ``{DELTACON:} {A} principled massive-graph similarity function,''
\newblock {\em CoRR}, vol. abs/1304.4657, 2013.

\bibitem{metrics}
M.~Tantardini, F.~Ieva, L.~Tajoli, and C.~Piccardi,
\newblock ``Comparing methods for comparing networks,''
\newblock {\em Scientific Reports}, vol. 9, 12 2019.

\bibitem{lamda-dist-bunke2006graph}
H.~Bunke, P.J. Dickinson, M.~Kraetzl, and W.D. Wallis,
\newblock {\em A Graph-Theoretic Approach to Enterprise Network Dynamics},
\newblock Progress in Computer Science and Applied Logic. Birkh{\"a}user
  Boston, 2006.

\bibitem{lamda-dist}
R.~Wilson and P.~Zhu,
\newblock ``A study of graph spectra for comparing graphs and trees,''
\newblock {\em Pattern Recognition}, vol. 41, pp. 2833--2841, 09 2008.

\end{thebibliography}

\end{document}